\documentclass[11pt,a4paper,reqno,twoside]{amsart}

\usepackage{epsfig}
\usepackage{xspace}
\usepackage{graphicx}
\usepackage{amssymb}
\usepackage{algorithm}
\usepackage[noend]{algpseudocode}
\usepackage{amsmath}
\usepackage{amsthm}
\usepackage{setspace}
\usepackage{graphicx}
\usepackage{mathptmx}      

\newtheorem{theorem}{Theorem}[section]
\newtheorem{lemma}[theorem]{Lemma}
\newtheorem{proposition}[theorem]{Proposition}

\newtheorem{definition}[theorem]{Definition}

\newtheorem{example}{\textbf{Example}}

\newcommand{\Z}{{\mathbb Z}}

\newcommand{\N}{{\mathbb N}}

\newcommand{\ie}{{\em i.e.},\xspace }

\begin{document}

\title{ An Algorithmic Characterization of Polynomial Functions over $\Z_{p^n} $}

\title
{An Algorithmic Characterization of Polynomial Functions over $\Z_{p^n} $}
\author{Ashwin Guha}
\address{Department of Computer Science and Automation \\ Indian Institute of
Science \\ Bangalore 560012, India.}
\email{guha\_ashwin@csa.iisc.ernet.in}

\author[A. Dukkipati]{Ambedkar Dukkipati}
\address{Department of Computer Science and Automation \\ Indian Institute of
Science \\ Bangalore 560012, India.}
\email{ambedkar@csa.iisc.ernet.in}

\begin{abstract}

In this paper we consider polynomial representability of functions defined 
over $\Z_{p^n}$, where $p$ is a prime and $n$ is a positive integer. 
Our aim is to provide an algorithmic characterization that (i) answers the decision 
problem: to determine whether a given function over $\Z_{p^n}$ is 
polynomially representable or not, and (ii) finds the polynomial if it is polynomially representable. 
The previous characterizations given by Kempner (1921)   
and Carlitz (1964) are existential in nature and only lead to 
an exhaustive search method, \ie algorithm with complexity  exponential in size of the input.
Our characterization leads to an algorithm whose running time is linear in size of input. 
We also extend our result to the multivariate case.


\end{abstract}

\maketitle

\section{Introduction}
\label{intro}

The problem of polynomial representability of functions is central to
many branches of mathematics. 
In literature, there have been attempts to represent various functions using 
polynomials and power series.
With the advent of calculus various 
methods were developed to approximate analytic functions using polynomials. An important milestone in this 
regard is the Taylor series, put forth by Brook Taylor. 

It is well known that if the underlying set is a finite field,
every function from the field to itself can be represented as a polynomial. 
The fact that every function over finite fields of the form $\Z_p$, where $p$ is prime, can be represented by a 
polynomial was noted by Hermite \cite{hermite1863fonctions}. 
Dickson  proved the above property 
for a general finite field in \cite{dickson1896analytic}. Dickson also showed that for a finite field of order $q$ every function is 
uniquely determined by a polynomial of degree less than $q$.
Polynomials over finite fields are also discussed in \cite{carlitz1977functions}. 
A comprehensive survey 
regarding finite fields can be found in \cite{lidl1997finite}.

In this paper we consider polynomial representability in $\Z_{p^n}$, where $p$ is a prime and $n$ is a positive
integer. Such residue rings have an elegant
structure and their study is the first step to understand polynomial representability in rings. 
This problem has been studied in literature and the two
important results were given by Carlitz \cite{carlitz1964functions} and Kempner \cite{kempner1921polynomials}.

A necessary and sufficient condition for a function over $\Z_{p^n}$
to be polynomial using Taylor series is provided in \cite{carlitz1964functions}.
Kempner \cite{kempner1921polynomials} showed that the only 
residue class rings where all functions can be represented by polynomials 
 are $\Z_p$, where $p$ is prime.
Kempner also provides a method to enumerate all polynomial functions over $\Z_t$ for any positive 
integer $t$.

A simpler formula to express the number of polynomial functions in $\Z_{p^n}$ is given in \cite{keller1968counting}. 
An alternative formula 
for the same is provided in \cite{mullen1984polynomial}, which is also extended to polynomials in several variables. 
 The formula is generalized over a 
Galois ring in \cite{brawley1992functions}. Some other related work 
can be found in \cite{zhang2004polynomial}.

Until now the problem of polynomial 
representability has been viewed from a traditional standpoint and its 
computational aspects have been ignored. In this paper we give an alternate 
characterization by considering the set of functions over $\Z_{p^n}$ as a 
$\Z_{p^n}$-module. We provide a linear time algorithm that solves the problem 
of polynomial representability and identify the polynomial which corresponds to the 
given function. Further, we give the characterization in the multivariate case.

This paper is organized as follows. In Section~\ref{Preliminaries} we provide 
the background and motivation for a new characterization. 
Section~\ref{IdentifyingPolynomialFunctions} contains our characterization of polynomial 
functions in $\Z_{p^n}$. In Section~\ref{AlgorithmicAspects} we give an algorithm 
based on our characterization. 
We also discuss its correctness and complexity. In Section~\ref{FindingPolyn} we 
determine the polynomial that corresponds to the given function. In Section~\ref{mutivariatecase} we extend the 
characterization to functions in several variables. 
Concluding remarks are provided in Section~\ref{Coda}.


\section{Background and Preliminaries}
\label{Preliminaries}

In this section we look at polynomials over finite rings, 
in particular, polynomials over residue class rings.
Let $t$ be a positive integer. One can easily verify that 
$\Z_t[x]$ is a $\Z_t$-module. Every polynomial over $\Z_t$ defines a function 
from $\Z_t$ to $\Z_t$ by the universal property of polynomial rings. In 
other words, if we allow the indeterminate $x$ to vary over $\Z_t$, 
then each polynomial corresponds to a mapping from $\Z_t$ to $\Z_t$. 
Let $\mathfrak{F}_t$ denote the set of all functions from $\Z_t$ to $\Z_t$. 
$\mathfrak{F}_t \cong (\Z_t)^t$, therefore we represent each element of 
$\mathfrak{F}_t$ as a $t$-tuple $(a_0,a_1,a_2,\ldots,a_{t-1})$, 
which corresponds to the function $f$ with $f(i)=a_i$. 
$\mathfrak{F}_t$ is a $\Z_t$-module of cardinality $t^t$. 

It should be noted that there are infinitely many polynomials in $\Z_t[x]$. 
Let $\mathfrak{P}_t$ denote the set of \emph{distinct} functions produced by $\Z_t[x]$. 
$\mathfrak{P}_t$ is finite and a subset of $\mathfrak{F}_t$. 

\begin{definition}
 A ring $A$ is said to be \textbf{polynomially complete} if every 
function from $A$ to itself can be represented as a polynomial.
\end{definition}

Examples of such rings are $\Z_2, \Z_3, \Z_5$.
In general, $\Z_p$ is polynomially complete, if $p$ is a prime. 
In other words, for prime $p$ we have $\mathfrak{P}_p\:=\:\mathfrak{F}_p$. 
Given any function 
it is possible to construct a polynomial which corresponds to 
that function. This is achieved using Lagrange interpolation which is possible because $\Z_p$ is a field. 
This does not hold for an 
arbitrary integer $t$. 
Polynomially complete 
structures are discussed extensively in \cite{lausch1973algebra}.

Kempner discusses polynomials over $\Z_t$ for any positive 
integer $t$ in \cite{kempner1921polynomials}. Kempner gives a method to compute the cardinality of $\mathfrak{P}_{t}$. The conditions for two 
polynomials to be equal as functions \ie $f(x) \equiv g(x)$ mod $t$ is described using 
the ideas of \emph{signature} and \emph{characteristic} of $t$.
A method to enumerate all distinct polynomial functions is also provided.

Carlitz \cite{carlitz1964functions} proved a key 
result regarding polynomial representation of functions in residue class 
ring modulo prime power. The result is very similar to Taylor series.  The result states that
a function $f$ from $\Z_{p^n}$ to itself is polynomial if and only if there exist 
functions $\Phi_0, \Phi_1,\ldots, \Phi_{n-1}$  over $ \Z_{p^{n}}$ such that for all $x,s \in \Z_{p^{n}}$, we have
\begin{equation}
\label{Carlitzeqn} 
f(x+sp)\equiv \Phi_0(x) + (sp)\Phi_1(x) + \ldots + (sp)^{n-1}\Phi_{n-1}(x) \mbox{ mod }{p^{n}}.
\end{equation}

A key feature, and in a certain respect, a drawback, is that these results 
use existential proofs. The results hinge on the existence of some functions 
satisfying certain properties. The previous works do not address the issue of 
finding the afore mentioned functions. Consequently these results do not lead 
to any constructive method to test whether a given function is polynomial 
representable, hence the results cannot be implemented in computation. 

One can apply a brute-force algorithm using the result by Carlitz, 
by considering all possible functions in 
$\mathfrak{F}_{p^{n}}$ as shown below. 

\begin{algorithmic}
\State \textbf{Input:} $f=(a_0,a_1,\ldots,a_{{p^{n}}-1})$.

\For {all $x$, $s \in \Z_{p^{n}}$}
 
	\For{ all $\Phi_0, \Phi_1, \ldots, \Phi_{n-1} \in \mathfrak{F}_{p^{n}}$}

	    \If {$ f(x+ps) = \Phi_0(x) + (sp)\Phi_1(x) + \ldots + (sp)^{n-1}\Phi_{n-1}(x)$}
		  \State \textbf{Output : } $f$ is polynomial
		  \State exit
	    \Else
		  \State \textbf{Output : } $f$ is not polynomial
	    \EndIf
	\EndFor
\EndFor

\end{algorithmic}

The above algorithm is extremely inefficient.  The cardinality of the ring 
$\Z_{p^{n}}$ is exponential in $p$. $|\mathfrak{F}_{p^{n}}|={p^{n}}^{p^{n}}$ is doubly exponential 
in $p$ making it infeasible to compute. 

We can modify the method in \cite{kempner1921polynomials} to suit our 
problem of testing whether a given function $f$ is polynomial. We can 
evaluate all polynomials in $\mathfrak{P}_{p^{n}}$ and compare it with $f$. 
If $f$ does not match any of the functions in $\mathfrak{P}_{p^{n}}$ we can infer 
that $f$ is not polynomially representable. 
The algorithm is presented below. 

\begin{algorithmic}
\State \textbf{Input:} $f=(a_0,a_1,\ldots,a_{{p^{n}}-1})$.

\For{ all $g \in \mathfrak{P}_{p^{n}}$ }
	\For{ all $x \in \Z_{p^{n}}$ }
	      \If{$f(x)=g(x)$}
		  \State \textbf{Output : } $f$ is polynomial
	      \Else
		  \State \textbf{Output : } $f$ is not polynomial
	      \EndIf
	\EndFor
\EndFor
\end{algorithmic}

This approach is better than the earlier one, 
still it is very inefficient. $\mathfrak{P}_{p^{n}}$, which is much smaller than $\mathfrak{F}_{p^{n}}$, 
is still extremely large. One can very well see that given an arbitrary function in 
$\mathfrak{F}_{p^{n}}$ it is less likely to be polynomially representable than otherwise.
A simple example from \cite{kempner1921polynomials} illustrates the magnitude of the sets involved.

\begin{example}
\label{exone}
Consider $p=3,n=11.$ Then,\\
${p^{n}}=3^{11} \approx 10^5$. \\
$|\mathfrak{F}_{p^{n}}| = 3^{11 \cdot 3^{11}} \approx 10^{1,000,000}$.\\
$|\mathfrak{P}_{p^{n}}| = 3^{3 \cdot 54} \approx 10^{76}$.
\end{example}

We can see that even for small values of $p$ and $n$, 
$\mathfrak{P}_{p^{n}}$ becomes unmanageably large. 

In this paper we provide an algorithm 
 that answers the question posed earlier. 
We present a new characterization to describe polynomial functions over $\Z_{p^n}$  
using which we bring down the complexity of the 
algorithm from doubly exponential in $p$ to exponential in $p$.

\section{Characterization of Polynomial Functions over $\Z_{p^n}$}
\label{IdentifyingPolynomialFunctions}

The question we wish to resolve can be stated as follows:
\begin{quote}
Given a prime $p$ and a positive integer $n$, and a function 
$f: \Z_{p^n} \longrightarrow\Z_{p^n}$, 
is there an algorithm to test whether $f$ is polynomially representable or not? 
\end{quote}

In order to answer the above question we make use of the module structure of $\mathfrak{P}_{p^n}$. 
One can easily verify that for any integer $t$, $\mathfrak{P}_{t}$, the set of polynomial 
functions, is a $\Z_{t}$-submodule of $\mathfrak{F}_t$.

\begin{lemma}
\label{ProposTwo}
 If $(a_0,a_1,\ldots,a_{t-1})$ and $(b_0,b_1,\ldots,b_{t-1})\in \mathfrak{P}_{t}$, then
$(a_0+b_0,a_1+b_1,\ldots,a_{t-1}+b_{t-1}) \in \mathfrak{P}_{t}.$
\end{lemma}
\proof
{
If $f$ and $g$ are the polynomials such that $f(x) = (a_0,\ldots,a_{t-1})$, 
and $g(x)=(b_0, \ldots, b_{t-1})$, for $x=0, \ldots, t-1$, then $h(x)$ defined as 
$f(x)+g(x)$ for all $x \in \Z_{t}$ is also a polynomial.

}

\begin{lemma}
\label{ProposThree}
If $(a_0,a_1,\ldots,a_{t-1}) \in \mathfrak{P}_{t}$, then $(sa_0,sa_1,\ldots,sa_{t-1}) \in \mathfrak{P}_{t}$,
where $s \in \Z_{t}$.
\end{lemma}
\begin{proof}
If $f(x) = (a_0,a_1,\ldots,a_{t-1})$, for $x=0, \ldots, t-1$, then $sf(x)$, which is also polynomial
corresponds to $(sa_0,sa_1,\ldots,sa_{t-1})$. Hence  $(sa_0,sa_1,\ldots,sa_{t-1}) \in \mathfrak{P}_{t}$.
\end{proof}

From these two lemmas we have the following proposition.
\begin{proposition}
$\mathfrak{P}_{t}$ is a $\Z_{t}$-submodule of $\mathfrak{F}_{t}$.
\end{proposition}

We intend to find a `suitable' 
generating set for $ \mathfrak{P}_{p^n}$ thereby 
translating it to a $\Z_{p^n}$-submodule membership problem.

\subsection{Paraphernalia}
\label{SectionThreeTwo}

\begin{definition}
\label{cyclicshiftdef}
Let $f \in \mathfrak{F}_{p^n}$. The $j$\textsuperscript{th} cyclic shift of $f$, denoted by $f^{<j>}$, is defined as 
$$ f^{<j>}(i)= f(i+j \: \mathrm{ mod } \: p^n)$$
for $i=0, \ldots, p^n-1. $
\end{definition}

\begin{definition}
 Let $v_1,v_2,\ldots,v_m \in \mathfrak{F}_{p^{n}}$. The $\Z_{p^n}$-submodule generated by $v_i$ for $i=1,2,\ldots,m$ and  
their cyclic shifts for $j=0, \ldots,p^n-1$ is denoted by \textbf{$\langle \langle v_1,v_2,\ldots,v_m \rangle \rangle $}.
\end{definition}

We shall identify a set $G' \subset  \mathfrak{P}_{p^{n}}$ 
such that $ \mathfrak{P}_{p^{n}} = \langle \langle G' \rangle \rangle$. 
The following lemma helps us describe such a set.

\begin{lemma}
\label{ProposOne}
 If $(a_0,a_1,\ldots,a_{{p^{n}}-1}) \in \mathfrak{P}_{p^{n}}$, then its cyclic shift 
$(a_1,\ldots,a_{{p^{n}}-1},a_0) \in \mathfrak{P}_{p^{n}}$.
\end{lemma}
\begin{proof}
 Let $f(x) \in \Z_{p^n}[x]$ be the polynomial that gives rise to the function\\
$(a_0,a_1,\ldots,a_{{p^{n}}-1})$. Then
$f(x+1)$, which is also a polynomial, gives rise to $(a_1,a_2,\ldots,a_{{p^{n}}-1},a_0)$.
Hence the cyclic shift also belongs to $\mathfrak{P}_{p^{n}}$.
\end{proof}

Clearly, shifting by $j$ places is equivalent to replacing $f(x)$ by $f(x+j)$. 
Hence all cyclic shifts are polynomially representable.
We now state and prove two lemmas which are crucial in establishing our main result.

\begin{lemma}
\label{LemmaU0}
 The function $u_0: \Z_{p^{n}} \longrightarrow \Z_{p^{n}}$ defined as\\ 
\begin{equation}
\label{U0eqn}
u_0(x)=	
\left\{
			\begin{array}{ll}
			     0 & \mbox{ if } p\nmid x \\
			     1 & \mbox{ if } p \:| \:x
			\end{array}
\right.
\end{equation}

belongs to $ \mathfrak{P}_{p^{n}}$.
\end{lemma}
\begin{proof}
We show that $u_0$ satisfies ~\eqref{Carlitzeqn}.
Let $\Phi_0(x) = u_0(x)$ and $\Phi_i$ be zero functions for $i=1,\ldots,n-1$.
Now if $p\nmid x$, $p\nmid (x+sp)$ for all $s \in \Z_{p^n}$. Therefore, 
\begin{eqnarray*}
u_0(x+sp) 
& = & 0\\
& = & u_0(x)+0\\
& = & \Phi_0(x) + (sp)\Phi_1(x) + (sp)^2\Phi_2(x) + \ldots + (sp)^{n-1}\Phi_{n-1}(x).
\end{eqnarray*}

If $p\:|\:x$ then $u_0(x)=1$ and $p\:|\:(x+sp)$ for all $s \in \Z_{p^{n}}$.
\begin{eqnarray*}
u_0(x+sp) 
& = & 1\\
& = & u_0(x)+0\\
& = & \Phi_0(x) + (sp)\Phi_1(x) + (sp)^2\Phi_2(x) + \ldots + (sp)^{n-1}\Phi_{n-1}(x).
\end{eqnarray*}

Hence $u_0 \in \mathfrak{P}_{p^{n}}$.

\end{proof}

\begin{lemma}
\label{LemmaUi}
The function $u_k: \Z_{p^{n}} \longrightarrow \Z_{p^{n}}$ defined as
\begin{equation}
\label{Uieqn}
 u_k(x)=	\left\{
			\begin{array}{ll}
			     0   & \mbox{ if } p\nmid x, \\
			     x^k & \mbox{ if } p \:|\: x.
			\end{array}
\right.
\end{equation}
belongs to $ \mathfrak{P}_{p^{n}}$ for $k=1,2,\ldots,n-1$.
\end{lemma}
\begin{proof}
We make use of ~\eqref{Carlitzeqn} again. 
Define $\Phi_0 = u_k$, for a fixed $k \in \{1,2,\ldots,n-1 \}$. For $i=1,2,\ldots,k$ define $\Phi_i$ as\\
$$\Phi_i(x)=	
\left\{
			\begin{array}{ll}
			     0   & \mbox{ if } p\nmid x \\
			     { k\choose i}  x^{k-i} & \mbox{ if } p \:| \:x
			\end{array}
\right.
$$

For $k<i \leq n-1$ define $\Phi_i$ as zero function.

If $p\nmid x$ then $u_k(x+ps)=u_k(x)=0$ and it satisfies ~\eqref{Carlitzeqn}.

If $p\:|\:x$ then 
\begin{eqnarray*}
u_k(x+sp) 
& = & (x+sp)^k\\
& = & x^k + {{k}\choose{1}}x^{k-1}(sp)+ {{k}\choose{2}}x^{k-2}(sp)^2+ \ldots+ {{k}\choose{k}}x^{0}(sp)^k\\
& = & u_k(x) + (sp)\Phi_1(x)+\ldots+ (sp)^k\Phi_k(x)+0\\
& = & \Phi_0(x) + (sp)\Phi_1(x) + (sp)^2\Phi_2(x) + \ldots + (sp)^{n-1}\Phi_{n-1}(x).
\end{eqnarray*}
Hence it satisfies ~\eqref{Carlitzeqn}. Therefore $u_k \in \mathfrak{P}_{p^{n}}$.

\end{proof}

Lemma~\ref{LemmaU0} is in fact a special case of Lemma~\ref{LemmaUi} when 
$k=0$. 
Lemma~\ref{LemmaUi} essentially means that the following vectors can 
be represented as polynomials.
\begin{eqnarray*}
u_0  & = & (1,\underbrace{0,\ldots,0}_\text{p-1 times},1,\underbrace{0,\ldots,0}_\text{p-1 times}, 1,\ldots,1,\underbrace{0,\ldots,0}_\text{p-1 times})\\
u_1  & = & (0,\underbrace{0,\ldots,0}_\text{p-1 times},p,\underbrace{0,\ldots,0}_\text{p-1 times},2p,\ldots,(p^n-p),\underbrace{0,\ldots,0}_\text{p-1 times})\\
u_2  & = &(0,\underbrace{0,\ldots,0}_\text{p-1 times},p^2,\underbrace{0,\ldots,0}_\text{p-1 times},(2p)^2,\ldots,(p^n-p)^2,\underbrace{0,\ldots,0}_\text{p-1 times})\\
&  & \qquad \qquad \qquad \qquad  \vdots \\
u_{n-1} & = &(0,\underbrace{0,\ldots,0}_\text{p-1 times},p^{n-1},\underbrace{0,\ldots,0}_\text{p-1 times},(2p)^{n-1},\ldots,(p^n-p)^{n-1},\underbrace{0,\ldots,0}_\text{p-1 times})
\end{eqnarray*}
\subsection{The Characterization}
\label{SectionThreeThree}

We now provide the main result of this paper. 
It asserts that a function is polynomial if and only if it belongs to 
the submodule generated by $u_k$ for $k=0, \ldots,n-1$ and their cyclic shifts.

\begin{theorem}
\label{Myresult}
 $f \in \mathfrak{P}_{p^{n}}$ if and only if
 $f \in \langle \langle u_0, u_1,\ldots, u_{n-1}\rangle \rangle$, 
where $u_k$ for $k=0, \ldots, n-1$ are defined as in ~\eqref{Uieqn} and 
$\langle \langle u_0, u_1,\ldots, u_{n-1}\rangle \rangle$ denotes the set generated by 
the vectors $u_k$ for $k=0, \ldots, n-1$ and their cyclic shifts.
\end{theorem}
\begin{proof}
($\Longrightarrow$)To show that $f \in \langle \langle u_0, u_1,\ldots, u_{n-1}\rangle \rangle \mbox{ implies } f \in \mathfrak{P}_{p^{n}} $.\\
From Lemma~\ref{LemmaUi} we know that $u_k \in \mathfrak{P}_{p^n}$
for $k=0, \ldots, n-1$. 
Let $u_{k}^{<j>}$ denote the $j$\textsuperscript{th} cyclic shift of $u_k$. From 
Lemma~\ref{ProposOne} we know that $u_{k}^{<j>} \in \mathfrak{P}_{p^{n}}$ for all $k=0,1,\ldots, n-1$ and 
$j=0,1,\ldots, {p}-1$. From Lemmas~\ref{ProposTwo} and ~\ref{ProposThree} we know that 
linear combinations of $u_{k}^{<j>} \in \mathfrak{P}_{p^{n}}$.

Let $ f \in \langle \langle u_0,\ldots,u_{n-1}\rangle \rangle$. Then there exist
scalars $\alpha_{k,j} \in \Z_{p^{n}}$, for $k=0,1,\ldots,n-1$ and $j=0,1,\ldots,p-1$
such that
\begin{align*}
 f = \: &\alpha_{0,0}u_0^{<0>} + \alpha_{0,1}u_{0}^{<1>} + \ldots + \alpha_{0,p-1}u_{0}^{<p-1>} \\
     &+ \alpha_{1,0}u_1^{<0>} + \alpha_{1,1}u_{1}^{<1>} + \ldots + \alpha_{1,p-1}u_{1}^{<p-1>} \\
     &\qquad \qquad \qquad \qquad \qquad \vdots \\
     &+ \alpha_{n-1,0}u_{n-1}^{<0>}+  \alpha_{n-1,1}u_{n-1}^{<1>}+ \ldots+ \alpha_{n-1,p-1}u_{n-1}^{<p-1>}\\
   = &\displaystyle\sum\limits_{k=0}^{n-1} \displaystyle\sum\limits_{j=0}^{p-1}  \alpha_{k,j}u_{k}^{<j>}.
\end{align*}
Clearly all terms in the summation belong to $\mathfrak{P}_{p^n}$. Hence $f \in \mathfrak{P}_{p^{n}}$.\\

\noindent($\Longleftarrow$)To show that $ f \in \mathfrak{P}_{p^{n}} \mbox{ implies } f \in \langle \langle u_0, u_1,\ldots, u_{n-1}\rangle \rangle$.\\
Let $f = (a_0,a_1,\ldots, a_{{p^{n}}-1})$, where $a_i \in \Z_{p^n}$, 
for $i=0, \ldots, p^n-1$. We can write $f$ as
$$f = v_0 + v_1+ \ldots + v_{p-1}, $$
where $v_j$ is the function defined as
\begin{equation}
\label{videfeqn}
v_j(i)= \begin{cases}
            a_i & \text{if } i \equiv j \text{ mod }p\\
	    0   & \text{if } i \not\equiv j \text{ mod }p,
           \end{cases}
\end{equation}
for $j=0, \ldots, p-1$. 
We now show that each $v_j \in \mathfrak{P}_{p^{n}}$. From ~\eqref{Carlitzeqn},
\begin{eqnarray*}
 a_{j+ps}
& = & f(j+ps)\\
& = & \Phi_0(j) + (sp)\Phi_1(j) + \ldots + (sp)^{n-1}\Phi_{n-1}(j).\\
\end{eqnarray*}
For $j=0,\ldots, p-1$
$$ v_j(i)= \begin{cases}
            \Phi_0(j) + (sp)\Phi_1(j) + \ldots + (sp)^{n-1}\Phi_{n-1}(j) & \text{if } i \equiv j \text{ mod }p\\
	    0   & \text{if } i \not\equiv j \text{ mod }p,
           \end{cases}
$$
where $i=j+ps.$ 
Each $v_j$ can be written as 
\begin{equation}
\label{asdf}
v_j = \eta_0^{(j)} + \ldots + \eta_{n-1}^{(j)},
\end{equation}
where $\eta_k^{(j)}$ denotes the function
$$ \eta_k^{(j)}(i)= \begin{cases}
            \Phi_k(j)(sp)^k & \text{if } i=j+sp\\
	    0   & \text{otherwise }
           \end{cases}
$$
for $k=0, \ldots, n-1$. From Lemma~\ref{LemmaUi}, we can see that 
$$\eta_k^{(j)} = \Phi_k(j)u_k^{<j>}. $$
From ~\eqref{asdf}
$$ v_j = \Phi_{0}(j)u_{0}^{<j>} + \ldots + \Phi_{n-1}(j)u_{n-1}^{<j>}.$$
$\Phi_k$ is well defined and $\Phi_k \in \Z_{p^{n}}$ for $k=0, 1, \ldots n-1$.
Hence $v_j$ is a linear combination of $u_k$, for $k=0, \ldots, n-1$ and their cyclic shifts.
Since $v_j \in\langle \langle u_0, u_1,\ldots, u_{n-1}\rangle \rangle$
for $j=0, \ldots, p-1$, $f \in \langle \langle u_0, u_1,\ldots, u_{n-1}\rangle \rangle$.

\end{proof}

Note that not all cyclic shifts of $u_k$ are required, 
for $k=0, \ldots, n-1$, but 
only the first $p$ shifts of each $u_k$. This is because all the other cyclic shifts 
can be written as linear combination of the first $p$ cyclic shifts. 
Hence each polynomial in $ \mathfrak{P}_{p^{n}}$ 
can be represented as a scalar sum of at most $np$ vectors.

This result is in fact a generalization of the generating set 
for vector space. The standard basis of the vector 
space $\mathfrak{F}_{p}$ corresponds to $u_0$ mentioned above 
and its cyclic shifts.

\section{Algorithm based on new characterization}
\label{AlgorithmicAspects}

Using Theorem~\ref{Myresult} we provide a method in Algorithm 1 which solves the decision problem mentioned earlier
by reducing it to a system of linear equations. The advantage of this reduction is that 
it is much easier to check if a system has solutions rather than check for 
the existence of functions which is done in ~\eqref{Carlitzeqn}. The
linear equations can be solved by standard computational methods.
We now present the algorithm based on the characterization. 
In the algorithm the following notations are used.

\noindent $A$ denotes the $(n-1) \times (n-1)$ matrix with elements from $\Z_{p^n}$

\begin{equation}
\label{reducedmatrix}
\left( 
\begin{array}{cccc}
p  &  p^2    & \ldots & p^{n-1}\\
2p & (2p)^2  & \ldots & (2p)^{n-1}\\
3p & (3p)^2  & \ldots & (3p)^{n-1}\\
\vdots & \vdots & &\vdots\\
(n-1)p& ((n-1)p)^2 & \ldots & ((n-1)p)^{n-1}
\end{array}
\right) .
\end{equation}
$v_i$ represents a $p^{n-1}$-tuple which forms a subarray of input for $i=0, \ldots, p-1$. 
$w_i$ represents a $p^{n-1}$-tuple of the form $(p^i,(2p)^i,\ldots,(p^n-p)^i)$ for $i=0, \ldots, n-1$ \ie
$u_i$ without the extraneous zeroes.\\

\begin{algorithm}
\label{modifiedalgo}
\caption{Determination of Polynomial Functions}
\begin{algorithmic}

\State \textbf{Input:} $f=(a_0,a_1,\ldots,a_{p^n-1})$, where $p$ is prime and $n \in \N$.\\

\State Split $f$ into $p$ subarrays $v_i$ such that \Comment Step 1\\
$v_i = (a_i, a_{i+p}, a_{i+2p},\ldots,a_{i+p^n-p}).$
\State 
\For {$i=0,1,\ldots,p-1$} \Comment Step 2
         \State $v_i= v_i - a_iw_0$\\
\EndFor 

\State Let $v_i=(0, b_1^{(i)}, b_2^{(i)}, \ldots, b_{p^{(n-1)}-1}^{(i)})$.
\State 
\For{$i=0,1,\ldots,p-1$}  \Comment Step 3
      \For{$j=0,1,\ldots,p^{(n-1)}-1$}
	    \If{$p \nmid b_{j}^{(i)}$}
		\State \textbf{Output:} $f$ is not polynomial.
		\State exit
	    \EndIf
      \EndFor
\EndFor

\State 
\For {$i=0,1,\ldots,p-1$} \Comment Step 4
         \If{$A
\left( 
\begin{array}{c}
x_1\\
x_2\\
\vdots\\
x_{n-1}            
\end{array}
\right) = 
\left( 
\begin{array}{c}
b_1^{(i)}\\
b_2^{(i)}\\
\vdots\\
b_{n-1}^{(i)}            
\end{array}
\right)$ has no solution} \Comment A as in ~\eqref{reducedmatrix}

	\State \textbf{Output:} $f$ is not polynomial.
	\State exit\\
\EndIf
\EndFor

\State Let $\Phi^{(i)} = (\Phi_{1}^{(i)},\Phi_{2}^{(i)},\ldots,\Phi_{n-1}^{(i)})$ be the solution.\\
 
\State 
\For{$i=0,1,\ldots,p-1$} \Comment Step 5
      \If{$v_i = \displaystyle\sum\limits_{j=1}^{n-1} \Phi_{j}^{(i)}w_i$}
\vspace{0.5cm}
	    \State \textbf{Output:} $f$ is polynomial.
      \Else
	    \State \textbf{Output:} $f$ is not polynomial.
      \EndIf
\EndFor

\end{algorithmic}
\end{algorithm}

One can see from ~\eqref{Carlitzeqn} that $f(x+sp)$ depends on $f(x)$. In step 1 of Algorithm 1 
we collect all the dependents in a single vector $v_i$ of length $p^{n-1}$. 
Note that all the $v_i$, for $i=0, \ldots, n-1$ are independent of each other.

Let $q=p^n$. Substituting $s=0$, we get $\Phi_0(x) = f(x)$ for all $x \in \Z_q$. 
In step 2 we subtract this first term from each $v_i$ to get a new vector
$$(0,a_{i+p}-a_i, a_{i+2p}-a_i, \ldots, a_{i+\frac{q}{p}}-a_i)$$ 
which is written as
$(0, b_1^{(i)}, b_2^{(i)}, \ldots, b_{\frac{q}{p}-1}^{(i)})$. 

~\eqref{Carlitzeqn} implies that 
if the input function $f$ is a polynomial then 
$a_{i+ps} - a_i$ must be divisible by $p$.
Therefore all $b_{j}^{(i)}$s must be zeroes or multiples 
of $p$. With a single pass on $v_i$, for $i=0, \ldots, n-1$, we perform this check in step 3. If any of the $v_i$ fails we 
conclude that $f$ is not polynomial. 

In step 4, we consider the following system of linear equations over $\Z_q$, 
with variables $x_i$.\\
\begin{equation}
\label{redsystem}
\left( 
\begin{array}{cccc}
p  &  p^2    & \ldots & p^{n-1}\\
2p & (2p)^2  & \ldots & (2p)^{n-1}\\
\vdots & \vdots & &\vdots\\
(n-1)p& ((n-1)p)^2 & \ldots & ((n-1)p)^{n-1}
\end{array}
\right)
\cdot
\left( 
  \begin{array}{c}
    x_1\\
    x_2\\
   \vdots\\
   x_{n-1}
  \end{array}
\right) = 
\left( 
  \begin{array}{c}
    b_1^{(i)}\\
    b_2^{(i)}\\
   \vdots\\
    b_{n-1}^{(i)}          
  \end{array}
\right)
\end{equation}

\noindent 
Here we make use of ~\eqref{Carlitzeqn} to check if there exist functions $\Phi_j$ such that 
\[ \begin{array}{rll}
a_{i+p} - a_i =& b_{1}^{(i)} =& p\Phi_1 + p^2\Phi_2 + \ldots + p^{n-1}\Phi_{n-1}  \\
a_{i+2p} - a_i= &  b_{2}^{(i)} =& 2p\Phi_1 + (2p)^2\Phi_2 + \ldots + (2p)^{n-1}\Phi_{n-1} \\
 & &\qquad \qquad \vdots \\
a_{i+(n-1)p} - a_i= &  b_{n-1}^{(i)}= & (n-1) p\Phi_1 + ((n-1)p)^2\Phi_2 + \ldots + ((n-1)p)^{n-1}\Phi_{n-1}
\end{array}
\] 

We remind ourselves that we are working with elements from the ring $\Z_{p^n}$, 
where division by $p$ is not defined. However, if $f$ happens to be a polynomial, 
then all multiples of $p$ in $b_j^{(i)}$ evenly cancel out. If at any stage  
a division by $p$ is encountered it immediately implies that $f$ is not polynomial, 
since the system has no solution.

If solution exists for all $i=0,1,\ldots,p-1$, we then proceed to check in step 5 if the solution satisfies the 
condition for remaining components of $v_i$, \ie we check if 
$$v_i \in \langle u_{0}^{<i>}, u_{1}^{<i>}, \ldots, u_{n-1}^{<i>} \rangle, $$
where $u_j^{<i>}$ is the $i$\textsuperscript{th} cyclic shift of $u_j$. If the above condition is true for 
$i=0,1,\ldots,p-1$ we conclude that $f$ is polynomial representable. 

The reason we choose to  check for the $(n-1)$ components first separately is because 
had we considered all the components together we would have arrived at an over-defined 
system of equations with $p^n-1$ equations for $n-1$ variables. Computation of rank to 
check for solutions would take $O((p^n)^2)$ instead of $O(n^2)$ as in the case of our algorithm.  

The Algorithm 1 can be fully understood with the help of an example.
\begin{example}
 Consider $p=2, n=3$. Then 

\[ \begin{array}{lcccccccc}
u_0=&(1 & 0 & 1 & 0 & 1 & 0 & 1 & 0)\\
u_1=&(0 & 0 & 2 & 0 & 4 & 0 & 6 & 0)\\
u_2=&(0 & 0 & 4 & 0 & 0 & 0 & 4 & 0)
\end{array}
\] 
\end{example}
Let $f$ over $\Z_8$ be defined as 
$$f= (2,1,6,1,2,1,6,1).$$

After Step 1:
\[
\begin{array}{lcccc} 
v_0=& (2 & 6 & 2 & 6)\\
v_1=& (1 & 1 & 1 & 1)
\end{array}
\]

After Step 2:
\[
\begin{array}{lcccc} 
v_0=& (0 & 4 & 0 & 4)\\
v_1=& (0 & 0 & 0 & 0)
\end{array}
\]
After Step 3 we find that all the entries are divisible by $2$.

\noindent In Step 4\\
For $v_0$:
\[
\left(
\begin{array}{cc}
 2 & 4\\
 4 & 0
\end{array}
\right)
\left(
\begin{array}{c}
 x_1\\ 
 x_2
\end{array}
\right)
=
\left(
\begin{array}{c}
 4 \\
 0
\end{array}
\right)
\]
for which solution exists, namely $x_1=2,x_2=0$. Hence, $v_0 \in \mathfrak{P}_8$.
We are misusing the notation slightly. We have avoided the 
extra zeroes for clarity.\\
Clearly $v_1 =(0,0,0,0) \in \mathfrak{P}_8$. Therefore $f$ is a 
polynomial function.

\begin{proposition}
 Algorithm 1 computes whether input function is
polynomially representable.
\end{proposition}

\begin{proof}
The proof of termination of the algorithm is trivial because  
of the finite nature of the structures involved.

From Theorem~\ref{Myresult} we have that 
a function $f$ is polynomial if and only if $f\in \langle \langle u_0, u_1,\ldots, u_{n-1}\rangle \rangle$. 
More specifically,
$f$ is   polynomial if and only if $v_i \in \langle u_0^{<i>}, u_1^{<i>},\ldots, u_{n-1}^{<i>}\rangle$, 
for $i=0,1,\ldots, p-1$, where $u_{j}^{<i>}$ denotes the $i$\textsuperscript{th} cyclic shift of $u_j$.

In other words, there exist scalars $\alpha_0, \alpha_1, \ldots, \alpha_{n-1}$
in $\Z_{p^{n}}$ such that
$$v_i= \alpha_0u_0^{<i>} + \alpha_1u_1^{<i>}+ \ldots + \alpha_{n-1}u_{n-1}^{<i>}.  $$

Suppose for convenience we drop the implicit zeros and write vector $v_i$ as
$v_i = (b_{1}^{(i)},b_{2}^{(i)},\ldots, b_{p^{(n-1)}-1}^{(i)})$, where $b_{j}^{(i)}$
are as described in Algorithm 1. 
Then 
there exist scalars $\alpha_0, \alpha_1, \ldots, \alpha_{n-1}$ such that
$$
\left( 
  \begin{array}{c}
    b_1^{(i)}\\
    b_2^{(i)}\\
   \vdots\\
    b_{p^{(n-1)}-1}^{(i)}     
  \end{array}
\right)
= 
\alpha_0
\left( 
  \begin{array}{c}
   1 \\
   1 \\
   \vdots\\
   1            
  \end{array}
\right)
+
\alpha_1
\left( 
  \begin{array}{c}
   p \\
   2p \\
   \vdots\\
   {p^{n}}-p           
  \end{array}
\right)+ \ldots +
\alpha_{n-1}
\left( 
  \begin{array}{c}
   p^{n-1} \\
   (2p)^{n-1} \\
   \vdots\\
   ({p^{n}}-p)^{n-1}            
  \end{array}
\right) .
$$\\

After step 2 of Algorithm 1, we get the first component of $v_i$ to be zero, \ie we eliminate the 
contribution of $u_0$. 
Let $\mathbf{x}$
be the vector $(x_1,x_2,\ldots,x_{n-1})$. We check for solutions of
\begin{equation}
\label{fullmatrix}
 \left( 
\begin{array}{cccc}
p  &  p^2    & \ldots & p^{n-1}\\
2p & (2p)^2  & \ldots & (2p)^{n-1}\\
\vdots & \vdots & &\vdots\\
(p^n-p)& (p^n-p)^2 & \ldots & (p^n-p)^{n-1}
\end{array}
\right)
\cdot
\left( 
  \begin{array}{c}
    x_1\\
    x_2\\
   \vdots\\
   x_{n-1}
  \end{array}
\right) = 
\left( 
  \begin{array}{c}
    b_1^{(i)}\\
    b_2^{(i)}\\
   \vdots\\
    b_{q/p}^{(i)}          
  \end{array}
\right)
\end{equation}
Now each
 $v_i \in \langle \langle u_0^{<i>}, u_1^{<i>}, \ldots, u_{n-1}^{<i>}  \rangle \rangle$ if
~\eqref{fullmatrix} has a solution for $x_i$ in $\Z_{p^{n}}$. 
Let $A$ be the matrix defined in ~\eqref{reducedmatrix}. 
Then,
$v_i \in \langle u_0^{<i>}, u_1^{<i>},\ldots, u_{n-1}^{<i>}\rangle$
if and only if $A \cdot \mathbf{x} = v_i$ in ~\eqref{redsystem} has solution  
and for $j=n, n+1, \ldots, p^{n-1}-1$
$$ b_{j}^{(i)} = \displaystyle\sum\limits_{j=1}^{n-1} \alpha_i(jp)^i.$$ 

Step 3 checks if $A \cdot \mathbf{x}$ in ~\eqref{redsystem} has a solution. 
Step 4 checks if the solution obtained in previous 
step satisfies for remaining components in ~\eqref{fullmatrix}.

\end{proof}

We now give a brief analysis of space and time complexities of the algorithm. 
We assume that the input is given in an array of size ${p^{n}}$, which is a reasonable 
assumption. Also we assume that addition and scalar multiplication on vector of 
size ${p^{n}}$ takes $O({p^{n}})$ time. 

\subsubsection*{Time complexity:}
Step 1 takes constant time as no explicit computation is involved: $O(1)$.\\
Step 2 involves a vector addition: $O({p^{n-1}})$.\\
Step 3 involves one array traversal: $O({p^{n-1}})$.\\
Step 4 involves computing rank of $(n-1) \times (n-1)$ matrix to check for solution. 
If solution exists it can be found using Gaussian elimination: $O(n^3)$.\\
Step 5 involves a comparison between two vectors: $O({p^{n-1}})$.\\

Note that steps 2-5 can be performed in parallel as the $v_i$ are independent 
of each other. Assuming a sequential model of computation,
\begin{eqnarray*}
\mbox{T}(p,n) 
& = & O(1)+ O({p^{n}})+ O({p^{n}}) + O(pn^3)+ O({p^{n}})\\
& = & O({p^{n}} + pn^3)\\
& = & O(p^n + pn^3).
\end{eqnarray*}

For all practical purposes $n^3 \ll p^n$. 
Hence time complexity is linear in size of input.

\subsubsection*{Space complexity:}
The input takes $O({p^{n}})$, which is unavoidable. 
Apart from that the only space requirement is to store the 
$(n-1) \times (n-1)$ matrix which takes $O(n^2)$. Hence space complexity is $O(n^2)$.

\section{Determination of the polynomial}
\label{FindingPolyn}
A natural continuation of the problem is to find the polynomial which corresponds to the given function. 
This can accomplished by merely giving the polynomials that correspond to the elements in the generating 
set. Improving upon Algorithm 1 we can obtain a solution of the system of linear equations, if it 
exists. Since the solution corresponds to the scalars in the linear combination of the generating elements, 
if we are equipped with the polynomials corresponding to the vectors $u_i$ defined in Lemma~\ref{LemmaUi} for $i=0, \ldots, n-1$, 
determining the polynomial of the given function becomes a trivial task. In this section we 
present the polynomials that correspond to the generating vectors.

\begin{proposition}
 The polynomial $(1-x^{\phi(p^n)})$ corresponds to the function $u_0$ defined in Lemma~\ref{LemmaU0} as 
\[
 u_0(x) =
  \begin{cases}
   0 & \text{if } p\nmid x \\
   1 & \text{if } p\:|\:x,
  \end{cases}
\] where $\phi(m)$ refers to Euler's totient function.
\end{proposition}
\begin{proof}
 From Euler's totient theorem we have
$$x^{\phi(m)} \equiv 1 \mbox{ mod } m, $$
for all $x$ such that $\mathrm{gcd}(x,m) =1$.

When $m =p^n$ we have
$$\phi(p^n)= p^{n} - p^{n-1},$$
$$x^{\phi(p^n)} \equiv 1 \mbox{ mod } p^n$$
if and only if $\mathrm{gcd}(x,p^n)=1$.

In other words we have $x^{\phi(p^n)} \equiv 1$ mod $p^n$ 
if $p \nmid x$ for all $ x \in \Z_{p^n}$. Also $\phi(p^n) > n$ for all 
$p \geq 2, n \geq 1$. Hence $(lp)^{\phi(p^n)} \equiv 0$ mod $p^n$, where $l \in \Z_{p^n}$, which 
means if $p\: | \:x$ then $x^{\phi(p^n)} \equiv 0$.

From these two observations we infer that the monomial $x^{\phi(p^n)}$ corresponds to the function 
\[
 x^{\phi(p^n)} =
  \begin{cases}
   1 & \text{if } p \nmid x\\
   0 & \text{if } p\:|\: x.
  \end{cases}
\]
Then the polynomial $(1-x^{\phi(p^n)}) \equiv (p^n -1)x^{\phi(p^n)} + 1$  corresponds to the function 
\[
 1-x^{\phi(p^n)} =
  \begin{cases}
   0 & \text{if } p\nmid x \\
   1 & \text{if } p\:|\:x 
  \end{cases}
\]
which is identical to the definition of $u_0$.

\end{proof}

It should be noted that many polynomials give rise to the function vector $u_0$. The polynomial 
mentioned above is just one of them. It is in fact possible to list all the polynomials which 
correspond to $u_0$ using the method given in \cite{kempner1921polynomials}.

Let $\mathfrak{u}_0$ denote the polynomial $1-x^{\phi(p^n)}$. Using $\mathfrak{u}_0$ one can easily 
construct the polynomials for all the generators of $\mathfrak{P}_{p^n}$. Each $u_i$ defined in 
Lemma~\ref{LemmaUi} as the function 
\[
 u_i(x)=	\left\{
			\begin{array}{ll}
			     0   & \mbox{ if } p\nmid x, \\
			     x^i & \mbox{ if } p \:|\: x
			\end{array}
\right.
\]
corresponds to 
the polynomial $\mathfrak{u}_i$ given as follows.
\[
 \mathfrak{u}_i =
  x^i\mathfrak{u}_0 =
\begin{cases}
  0 & \text{if } p\nmid x \\
  x^i & \text{if } p\:|\:x 
\end{cases}
=u_i.
\]

The cyclic shifts of $u_i$ are obtained by replacing $x$ by $x+j$ in each $\mathfrak{u}_i$. The polynomials 
corresponding to the generators are

\begin{align}
u_0  & \equiv 1-x^{\phi(p^n)}      \\
u_i  & \equiv x^i(1-x^{\phi(p^n)})     \\
u_i^{<j>} & \equiv (x+j)(1-(x+j)^{\phi(p^n)})
\end{align}
for $i=1, \ldots, n-1$ and $j=1, \ldots, p-1$.

Written explicitly the desired polynomials are \\
$ 1-x^{\phi(p^n)},  1-(x+1)^{\phi(p^n)}, \ldots,  1-(x+p-1)^{\phi(p^n)}, \\
x(1-x^{\phi(p^n)}), (x+1)(1-(x+1)^{\phi(p^n)}), \ldots, (x+p-1)(1-(x+p-1)^{\phi(p^n)}), \ldots,\\
x^{n-1}(1-x^{\phi(p^n)}), (x+1)^{n-1}(1-(x+1)^{\phi(p^n)}), \ldots, (x+p-1)^{n-1}(1-(x+p-1)^{\phi(p^n)}).
 $

\section{Polynomials in several variables}
\label{mutivariatecase}
The problem of determining whether a given function is polynomial can be extended to functions over several 
variables as well, \ie given a function $f : (\Z_{p^n})^m \longrightarrow \Z_{p^n}$, where $m$ is a positive integer, 
can we determine whether $f$ can be written as a polynomial in $m$ variables?
The characterization given in ~\eqref{Carlitzeqn} is extended to multivariate functions 
in \cite{carlitz1964functions}. As in the case of single variable the characterization is existential in nature. 
Some related work can be found in \cite{kaiser1987permutation}. We show 
that our characterization Theorem~\ref{Myresult} can be extended to multivariate functions.

Let $\mathfrak{F}_{p^n}^{(m)}$ denote the set of all functions from 
$(\Z_{p^n})^m$ to $\Z_{p^n}$. Let $\mathfrak{P}_{p^n}^{(m)}$ denote those 
functions which are polynomially representable.

The definition of cyclic shift in mulitvariate case is non-trivial, but 
follows closely the univariate case given in Definition~\ref{cyclicshiftdef}. 

\begin{definition}
Let $f \in \mathfrak{F}_{p^n}^{(m)}$. For $(j_1, \ldots, j_m) \in \Z^m$ we denote its cyclic shift
by $f^{<j_1, \ldots, j_m>}$ and define it as
$$
 f^{<j_1, \ldots, j_m>}(x_1, \ldots, x_m) = f(
 x_1 + j_1 \mbox{ mod } p^n, \ldots,
 x_m + j_m \mbox{ mod } p^n)
$$
for all $(x_1, \ldots, x_m) \in (\Z_{p^n})^m.$
\end{definition}

For functions involving several variables the result in \cite{carlitz1964functions} given in 
~\eqref{Carlitzeqn} takes the form: 
$f:  (\Z_{p^n})^m \longrightarrow \Z_{p^n}$ is polynomial if and only if there exists suitable 
functions $\Phi_{i_1,\ldots,i_m} : (\Z_{p^n})^m \longrightarrow \Z_{p^n}$ such that
$$ f(x_1+ ps_1, \ldots, x_m+ps_m)= 
\displaystyle\sum\limits_{i_1+\ldots+i_m <n} 
\Phi_{i_1,\ldots,i_m}(x_1,\ldots,x_m) (ps_1)^{i_1}\ldots(ps_m)^{i_m} \mbox{ mod }p^n. $$

Using the above result we define a generating set similar to the one defined earlier in Lemma~\ref{LemmaUi}.

\begin{lemma}
The function $u_{k_1, \ldots,k_m}: (\Z_{p^{n}})^m \longrightarrow \Z_{p^{n}}$ defined as
\begin{equation*}
\label{multiUieqn}
 u_{k_1,\ldots,k_m}(x_1,\ldots,x_m)=	
\left\{
			\begin{array}{ll}
			  x_1^{k_1}\ldots x_m^{k_m} & \mbox{ if } p \:|\: x_i \mbox{ for all }i=1,\ldots,m\\  
			  0   & \mbox{ if } p\nmid x_i \mbox{ for at least one }i=1,\ldots,m 
			     
			\end{array}
\right.
\end{equation*}
belongs to $ \mathfrak{P}_{p^{n}}^{(m)}$, where $0 \leq k_1, \ldots, k_m < n$.
\end{lemma}
\begin{proof}
Proof is by induction on $m$. For $m$=1, the above statement is true by Lemma~\ref{LemmaUi}. 
Assume that the function
$u_{k_1, \ldots,k_{m-1}}: (\Z_{p^{n}})^{m-1} \longrightarrow \Z_{p^{n}}$ defined as
\begin{equation*}
 u_{k_1,\ldots,k_{m-1}}(x_1,\ldots,x_{m-1})=	
\left\{
			\begin{array}{ll}
			  x_1^{k_1}\ldots x_{m-1}^{k_{m-1}} & \mbox{ if } p \:|\: x_i \mbox{ for all }i=1,\ldots,m-1.\\  
			  0   & \mathrm{ otherwise, }
			\end{array}
\right.
\end{equation*}
is polynomially representable. Let $h \in \Z_{p^n}[x_1,\ldots,x_{m-1}]$ be the polynomial 
which when evaluated over $(\Z_{p^n})^{m-1}$ gives the function $u_{k_1, \ldots,k_{m-1}}$.

Consider the function $u'_{k_m}: \Z_{p^n} \longrightarrow \Z_{p^n}$ defined as 
\begin{equation*}
 u'_{k_m}(x)=	\left\{
			\begin{array}{ll}
			     0   & \mbox{ if } p\nmid x, \\
			     x^{k_m} & \mbox{ if } p \:|\: x.
			\end{array}
\right.
\end{equation*}
From Lemma~\ref{LemmaUi} we know that it is polynomially representable.  
Let $g \in \Z_{p^n}[x_m]$ be the polynomial that corresponds to the function $u'_{k_m}$.

Consider $h$ and $g$ as polynomials in $\Z_{p^n}[x_1, \ldots, x_m]$. Clearly $hg \in \Z_{p^n}[x_1, \ldots, x_m]$.
Let $f=gh$. As a function $f$ is defined as 
\[
 f(x_1,\ldots,x_m) =
\begin{cases}
x_1^{k_1}\ldots x_{m-1}^{k_{m-1}} \cdot x_m^{k_m} & \text{ if } p\:| \: x_i \text{ for all }i=1, \ldots,m-1 \text{ and } p\:| \:x_m\\
x_1^{k_1}\ldots x_{m-1}^{k_{m-1}} \cdot 0& \text{ if }  p\:| \: x_i \text{ for all }i=1, \ldots,m-1 \text{ and } p\nmid x_m\\
 0 \cdot x_m^{k_m} & \text{ if }  p\nmid x_i \text{ for some }i=1, \ldots,m-1 \text{ and } p\:|\: x_m\\
 0 \cdot 0   & \text{ if } p\nmid x_i \text{ for some }i=1, \ldots,m-1 \text{ and } p\nmid x_m.
  \end{cases}
\]
That is we have 
\[
 f(x_1,\ldots,x_m) =
\begin{cases}
x_1^{k_1}\ldots x_m^{k_m} & \text{ if } p\:| \: x_i \text{ for all }i=1, \ldots,m,\\
 0   & \text{ if } p\nmid x_i \text{ for some }i=1, \ldots,m,
  \end{cases} 
\]
which is identical to $u_{k_1,\ldots,k_m}(x_1,\ldots,x_m)$. 
Hence $u_{k_1,\ldots,k_m}(x_1,\ldots,x_m) \in \mathfrak{P}_{p^n}^{(m)}$.

\end{proof}

\begin{theorem}
 $f \in \mathfrak{P}_{p^{n}}^{(m)}$ if and only if
 $f \in \langle \langle u_{k_1,\ldots,k_m} : k_1 + \ldots + k_m <n \rangle \rangle$, 
where $(k_1, \ldots, k_m) \in \Z^n$ and  $u_{k_1,\ldots,k_m}$ is defined as above.
\end{theorem}

\begin{proof}
The proof is similar to one given in Theorem~\ref{Myresult}. One implication is trivial. 
To prove that  $f$ is polynomial implies $f \in \langle \langle u_{k_1,\ldots,k_m}: k_1 + \ldots + k_m <n\rangle \rangle$ 
write $f$ as 
$$f= \displaystyle\sum\limits_{0 \leq j_1, \ldots, j_m < p} v_{j_1, \ldots,j_m}, $$ 
such that 
$v_{j_1, \ldots,j_m} : (\Z_{p^n})^m \longrightarrow \Z_{p^n}$ defined as 
\[
v_{j_1, \ldots,j_m}(a_1,\ldots,a_m) =
\begin{cases}
f(a_1, \ldots,a_m) & \text{ if } a_i \equiv j_i \text{ mod }p \text{ for all }i=1, \ldots,m\\
 0    & \text{ otherwise. }
  \end{cases}
\]

That is $f(a_1, \ldots, a_m)$ is placed in exactly one of the $p^m$ different $v_{j_1, \ldots,j_m}$ 
functions. We now show that  each $v_{j_1, \ldots,j_m} \in \mathfrak{P}_{p^n}^{(m)}$ for all 
$(j_1, \ldots, j_m) \in \Z^m$.

This can be written as 
\[
v_{j_1, \ldots,j_m}(a_1,\ldots,a_m)  \\
= 
\begin{cases}
\displaystyle\sum\limits_{k_1+\ldots + k_m < n} \Phi_{k_1,\ldots,k_m}(j_1,\ldots,j_m)
\displaystyle\prod\limits_{i=1}^{m} (ps_i)^{k_i}&  
 \text{ if } a_i = j_i + ps_i \\
      & \text{ for all }i=1, \ldots, m \\
 0    & \text{ otherwise. }
  \end{cases}
\]

Let $\eta_{k_1,\ldots,k_m} : (\Z_{p^n})^m \longrightarrow \Z_{p^n}$ such that\\
$$ v_{j_1, \ldots,j_m} = \sum_{k_1+\ldots + k_m < n}\eta_{k_1,\ldots,k_m}, $$
where 
$$
\eta_{k_1,\ldots,k_m}(a_1,\ldots,a_m)=
\begin{cases}
 \Phi_{k_1,\ldots,k_m}(j_1,\ldots,j_m)(ps_1)^{k_1}\ldots(ps_m)^{k_m}& 
 \text{ if } a_i = j_i + ps_i \\
      & \text{ for all }i=1, \ldots,m\\
 0    & \text{ otherwise. }
  \end{cases}
$$
for all $(a_1, \ldots,a_m)\in (\Z_{p^n})^m.$ We have 

$$\eta_{k_1,\ldots,k_m} = \Phi_{k_1,\ldots,k_m}(j_1,\ldots,j_m)\cdot  
 u_{k_1, \ldots, k_m}^{<j_1,\ldots,j_m>},$$
where $u_{k_1, \ldots, k_m}^{<j_1,\ldots,j_m>}$ denotes the $(j_1,\ldots,j_m)$ cyclic shift of 
$u_{k_1, \ldots, k_m}$. 
Hence each $\eta_{k_1,\ldots,k_m} \in \langle u_{k_1, \ldots, k_m}^{<j_1,\ldots,j_m>} : k_1 + \ldots + k_m <n \rangle$. 
In other words,
$ \eta_{k_1,\ldots,k_m} \in \langle \langle u_{k_1, \ldots, k_m}: k_1 + \ldots + k_m <n \rangle \rangle,$
which implies $v_{j_1, \ldots,j_m}$ and therefore 
$f \in \langle \langle u_{k_1, \ldots, k_m} : k_1 + \ldots + k_m <n \rangle \rangle$.

\end{proof}
Using the above result we can obtain an algorithm similar to Algorithm 1 
that determines whether the given function is polynomial or not. 
The complexity in multivariate case is $O((np)^m)$, which is linear in the size of the input.

Determination of the polynomial is extended to the multivariate case in a natural way. In the case of $m$ variables we know 
that $\mathfrak{P}_{p^n}^{(m)}$ is generated by $\{u_{k_1,\ldots,k_m} | \:k_1 + \ldots + k_m <n \}$. 
The function $u_{0,\ldots,0}$ is given by the polynomial 
$$(1-x_1^{\phi(p^n)})(1-x_2^{\phi(p^n)})\ldots (1-x_m^{\phi(p^n)}).$$

In general the function vector 
\[
u_{k_1,\ldots,k_m} =
  \begin{cases}
   x_1^{k_1}\ldots x_m^{k_m} & \text{if } p\:|\:x_i \text{ for all }i=1, \ldots, m \\
   0 & \text{otherwise, }
  \end{cases} 
\]
is given by the polynomial
$$x_1^{k_1}x_2^{k_2}\ldots x_m^{k_m}(1-x_1^{\phi(p^n)})(1-x_2^{\phi(p^n)})\ldots (1-x_m^{\phi(p^n)}). $$
This way it is possible to determine the polynomial that corresponds to the function in 
multivariate case as well.


\section{Concluding Remarks}
\label{Coda}

In this paper we considered the problem of polynomial representability 
of functions over $\Z_{p^n}$. A new characterization of polynomial functions 
is given that leads to a non-exhaustive algorithm which runs in linear time. 
We have also given a method to identify the polynomial that corresponds to the 
given function by providing the polynomials for the generating vectors. The 
results are extended to multivariate case as well.


\bibliographystyle{plain}     

\bibliography{ashwins}  

\end{document}